\documentclass[a4paper,USenglish]{lipics-v2018}

\usepackage{microtype}
\nolinenumbers
\usepackage{tikz}
\usepackage{mathtools}
\newtheorem{proposition}[theorem]{Proposition}

\author{Stéphane Le Roux}{LSV, CNRS \& ENS Cachan, Université Paris-Saclay, Cachan, France}{}{}{}

\author{Arno Pauly}{Swansea University, United Kingdom}{}{}{}

\author{Mickael Randour}{F.R.S.-FNRS \& UMONS -- Université de Mons, Belgium}{}{}{F.R.S.-FNRS Research Associate.}

\authorrunning{S.~Le Roux and A.~Pauly and M.~Randour}
\Copyright{CC-BY}

\subjclass{Formal languages and automata theory}


\title{Extending finite-memory determinacy by Boolean combination of winning conditions}
\titlerunning{Extending finite-memory determinacy by Boolean combination of winning conditions}

\funding{Research partially supported by F.R.S.-FNRS under Grant n° F.4520.18 (\textit{ManySynth}).}

\usetikzlibrary{positioning}
\usetikzlibrary{shapes.geometric}
\usetikzlibrary{automata}

\EventEditors{John Q. Open and Joan R. Access}
\EventNoEds{2}
\EventLongTitle{42nd Conference on Very Important Topics (CVIT 2016)}
\EventShortTitle{CVIT 2016}
\EventAcronym{CVIT}
\EventYear{2016}
\EventDate{December 24--27, 2016}
\EventLocation{Little Whinging, United Kingdom}
\EventLogo{}
\SeriesVolume{42}
\ArticleNo{20}
\hideLIPIcs

\begin{document}
\maketitle

\begin{abstract}
We study finite-memory (FM) determinacy in games on finite graphs, a central question for applications in controller synthesis, as FM strategies correspond to implementable controllers. We establish general conditions under which FM strategies suffice to play optimally, even in a broad multi-objective setting. We show that our framework encompasses important classes of games from the literature, and permits to go further, using a unified approach. While such an approach cannot match ad-hoc proofs with regard to tightness of memory bounds, it has two advantages: first, it gives a widely-applicable criterion for FM determinacy; second, it helps to understand the cornerstones of FM determinacy, which are often hidden but common in proofs for specific (combinations of) winning conditions.
\end{abstract}

\section{Introduction}
\label{sec:intro}

\smallskip\noindent\textbf{Controller synthesis through the game-theoretic metaphor.} \textit{Two-player games on graphs} are widely studied, notably for their applications in controller synthesis for reactive systems~\cite{DBLP:conf/dagstuhl/2001automata,rECCS,DBLP:conf/lata/BrenguierCHPRRS16,DBLP:reference/mc/BloemCJ18}. In this context, Player~1 is seen as the system to control, Player~2 as its uncontrollable environment, and the game models their interaction. The objective of Player~1 is to enforce a given specification represented as a \textit{winning condition}. The goal of synthesis is thus to decide if Player~1 has a \textit{winning strategy}, i.e., one that guarantees victory against all possible strategies of Player~2, and to build such a strategy efficiently if it exists.

Winning strategies are essentially \textit{formal blueprints} for controllers that one may want to implement in practical applications. With that in mind, the complexity of such strategies is of tremendous importance: the simpler the strategy, the easier and cheaper it will be to build the corresponding controller and maintain it. That is why a lot of research effort is put into pinpointing the complexity (in terms of memory and/or randomness) of strategies needed to play optimally for each specific class of games and winning conditions.

\smallskip\noindent\textbf{Memoryless determinacy.} An elegant result by Gimbert and Zielonka established more than ten years ago characterizes the winning conditions that enjoy \textit{memoryless determinacy}\footnote{The existence, in each vertex of the game, of a winning strategy that requires no memory at all, for one of the two players.} in two-player turn-based zero-sum games on finite graphs~\cite{DBLP:conf/concur/GimbertZ05}. These include conditions such as \textit{parity}, \textit{mean-payoff}, or \textit{energy}, all in the single-objective case.

\smallskip\noindent\textbf{The need for memory.} Over the last decade, the increasing need to model complex spec\-ifi\-ca\-tions has shifted research toward games where multiple quantitative and qualitative objectives co-exist and interact, requiring the analysis of \textit{interplay} and \textit{trade-offs} between several objectives. Consider a computer server responding to requests. Decreasing its response-time (modeled using a mean-payoff condition) would usually require to increase its computational power and energy consumption (modeled by an energy condition). Hence, we need to consider games where winning conditions are actually conjunctions of conditions, or even richer Boolean combinations. In this context, memoryless strategies do not suffice, and one has to use an amount of memory which can quickly become an obstacle to implementation (e.g., exponential memory) or which can prevent it completely (infinite memory).

See for example~\cite{DBLP:journals/tcs/ChatterjeeD12,DBLP:journals/acta/ChatterjeeRR14,DBLP:conf/icalp/JurdzinskiLS15} for combinations of energy and parity,~\cite{DBLP:journals/iandc/VelnerC0HRR15} for combinations of mean-payoff,~\cite{DBLP:journals/acta/BouyerMRLL18,DBLP:conf/fossacs/BouyerHMR017} for combinations of energy and \textit{average-energy} conditions,~\cite{DBLP:journals/iandc/Chatterjee0RR15} for combinations of \textit{total-payoff}, or~\cite{DBLP:journals/iandc/Chatterjee0RR15,DBLP:journals/corr/BruyereHR16} for combinations of \textit{window} objectives. Establishing precise complexity bounds for such general combinations of winning conditions is tricky and sometimes counterintuitive. For example, while energy games and mean-payoff games are inter-reducible in the single-objective setting, exponential-memory strategies are both sufficient and necessary for conjunctions of energy conditions while \textit{infinite-memory} strategies are required for conjunctions of mean-payoff ones.

\smallskip\noindent\textbf{Our contribution.} Our goal is to provide a \textit{general and abstract theorem} that lets us draw conclusions over games with complex objectives, provided that the primitive objectives used in the construction of the complex objectives fulfill some criteria.
Such an abstract approach provides results for many concrete types of objectives at once, and can be applied to new types of objectives in which the community may become interested later. Another advantage of an abstract approach is that it reveals partially the \textit{core features} determining whether the results hold or not. Admittedly, a downside of the abstract approach is that the concrete bounds that we obtain for the required memory will often be far worse than those established in concrete instances (as we depend on the most general upper bound).

Let us sketch our approach intuitively. First, we define the notion of \textit{regularly-predictable} winning condition: a winning condition is regularly-predictable if for every game using it, a finite automaton suffices to recognize histories from which Player~1 has a winning strategy. As we will show, many well-behaved winning conditions fall in this category (including all prefix-independent ones). Second, we consider \textit{regular languages} which subsume many simple winning conditions (e.g., fully bounded energy, window objectives). Third, we introduce the notion of \textit{hypothetical subgame-perfect equilibrium} (hSPE) that is a key technical tool for the theorem to come.\footnote{Morally, we consider regularly-predictable winning conditions for which FM strategies suffice, but for technical reasons, we need the slightly stronger assumption of existence of FM hSPE.} Finally, we prove our main result: given a class $\mathcal{W}$ of regularly-predictable winning conditions admitting FM hSPE and closed under Boolean combination, any winning condition obtained by Boolean combination of elements of $\mathcal{W}$ and regular languages is itself a regularly-predictable condition admitting FM hSPE.

\smallskip\noindent\textbf{Discussion.} The sketch depicted above may seem unsurprising to the expert reader. Indeed, regular languages are morally equivalent to \textit{safety} conditions, and combining reasonably well-behaved objectives with safety objectives should preserve FM determinacy (with a blow-up due the possible compact encoding of the safety-like conditions). This is perfectly true but in our opinion the interest of our contribution is elsewhere. First, while it is \textit{morally} easy to believe that such an extension of FM determinacy will hold, it is actually quite tedious to manage to prove it for general classes of games (e.g., see all aforementioned articles on specific combinations of objectives). Here, we provide a \textit{general proof} of this scheme of extension, based on (in our humble opinion) \textit{elegant and natural concepts} such as regularly-predictable objectives and hSPE. Second, we provide a thorough discussion of the hypotheses needed for our main theorem, and we show that they are relatively \textit{tight} in the sense that weakening any condition almost immediately leads to falsification of FM determinacy for the resulting complex objective. Thus, we give a \textit{rather complete picture of the frontiers of FM determinacy for combination of objectives} that, hopefully, will help in understanding its cornerstones and drive further research toward such amenable objectives.

\smallskip\noindent\textbf{Related work.} We already mentioned several papers studying specific (combinations of) winning conditions. Here, we highlight works where similar approaches have been considered to establish ``meta-theorems'' applying to general classes of games. Arguably the most important result in this direction is the determinacy theorem by Martin that guarantees determinacy (without considering the complexity of strategies) for Borel winning conditions~\cite{martin_AM75}. We already discussed Gimbert and Zielonka's focus on memoryless strategies~\cite{DBLP:conf/concur/GimbertZ05}. Kopczynski studied games where memoryless strategies suffice for only one of the players whereas his opponent requires memory~\cite{DBLP:conf/icalp/Kopczynski06}. Finally, a similar approach has been pursued in \cite{DBLP:journals/iandc/RouxP18}, where abstract criteria were identified that enable moving from FM determinacy of two-player win/lose games to the existence of FM Nash equilibria in multi-player multi-outcome games.

This paper is an extended version of the corresponding conference article~\cite{fsttcs}.

\smallskip\noindent\textbf{Outline.} In Sect.~\ref{sec:prelim}, we define the core concepts, discuss known results from the literature, and present an illustrative example that requires Boolean combinations, not only conjunctions. In Sect.~\ref{sec:thm}, we present the main theorem sketched above. In Sect.~\ref{sec:app}, we discuss classical (combinations of) objectives from the literature and how they fit (or not) in our framework. Then, in Sect.~\ref{sec:req}, we come back to the theorem to discuss its requirements and its relative tightness. Finally, in Sect.~\ref{sec:conj}, we present several more precise results that can be obtained from our approach when restricting combinations of objectives to conjunctions and disjunctions only.

\section{Preliminaries and example}
\label{sec:prelim}

\subsection{Definitions}

\begin{definition}
We fix a set $C$, calling its elements \emph{colors}. A $C$-colored \emph{arena} $O$ is a tuple $O = \langle V_1, V_2, E, \Gamma\rangle$, where $(V_1 \cup V_2,E)$ is a directed graph such that each vertex has at least one outgoing edge, and $\Gamma\colon V_1 \cup V_2 \to C$ is the coloring function. We assume \textit{finite} arenas.
\end{definition}

As we generally consider the set $C$ of colors fixed, we suppress $C$-colored in the following. A \emph{winning condition} is a set $W \subseteq C^\omega$ of infinite sequences of colors. An arena $O$ together with a winning condition $W$ constitutes a \emph{game} $g = (O, W)$.

A \emph{strategy} for Player $i$ is a function $\sigma_i\colon C^\ast \times V_i \to (V_1 \cup V_2)$ satisfying that for all $w \in C^*$ and $v \in V_i$ we have $(v,\sigma_i(w,v)) \in E$. This means that a strategy selects a successor vertex based on the current vertex \emph{owned} by the relevant player, and a finite sequence of colors. 
It is \emph{finite-memory} if it can be encoded by a deterministic \emph{Moore machine} $(M, m_0, \alpha_u, \alpha_n)$ where $M$ is a finite set of states (the memory of the strategy), $m_0 \in M$ is the initial memory state, $\alpha_u\colon M \times C \to M$ is the update function, and $\alpha_n\colon M \times V_i \to (V_1 \cup V_2)$ is the next-action function. The Moore machine defines a strategy $\sigma_i$ such that $\sigma_i(w v) = \alpha_n(\widehat{\alpha}_u(m_0,w),v)$ for all history $w \in C^\ast$ and vertex $v \in V_i$, where $\widehat{\alpha}_u$ extends $\alpha_u$ to sequences of colors as expected. The \textit{size} of the strategy is the size $|M|$ of its Moore machine. Note that a strategy is \textit{memoryless} when $|M| = 1$.

\begin{remark}
Let us comment on our use of colors. If we consider a fixed game only, we could w.l.o.g.~consider each vertex to be colored by itself. However, separating vertices and colors is necessary to speak about using the same winning condition together with different arenas.
Then one might expect that a strategy can take into account the history of vertices leading up to the current position, not merely the sequence of colors. However, it is straightforward to verify that players have no incentive to take these specifics into account. Moreover, the notion of hSPE defined below takes into account also histories not actually realizable in the given arena. Defining strategies on color histories is thus the most convenient approach.
\end{remark}

A \emph{strategy profile} is a pair of strategies $(\sigma_1,\sigma_2)$, which we identify with the function $\sigma = \sigma_1 \cup \sigma_2\colon C^* \times (V_1 \cup V_2) \to (V_1 \cup V_2)$. We say that Player $i$ can deviate from $(\sigma_1,\sigma_2)$ to $(\sigma'_1,\sigma'_2)$ if $\sigma_{3-i} = \sigma'_{3-i}$. A strategy profile $\sigma$ induces a function $\Sigma_\sigma\colon C^* \times (V_1 \cup V_2) \to C^\omega$ coinductively as follows: $\Sigma_\sigma(w,v)(i) = w(i)$ (where $w(i)$ represents the $i$th symbol of $w$) for $i \leq |w|$ and $\Sigma_\sigma(w,v)(k) = \Sigma_\sigma(w\Gamma(v),\sigma(w,v))(k)$ for $k = |w|+1$. The infinite color sequence $\Sigma_\sigma(w,v)$ is the \emph{play} starting at $v$ with history $w$.
For a play $\rho$, let $Pref(\rho)$ be the set of finite prefixes of $\rho$.

\begin{definition}
A \textit{hypothetical subgame-perfect equilibrium} (hSPE) is a strategy profile $\sigma$ such that for all  $w \in C^*$, $v \in V_1 \cup V_2$ the following implications hold: If $\Sigma_\sigma(w,v) \in W$, then $\Sigma_{\sigma'}(w,v) \in W$ for each $\sigma'$ that Player $2$ can deviate to from $\sigma$. If $\Sigma_\sigma(w,v) \notin W$, then $\Sigma_{\sigma''}(w,v) \notin W$ for each $\sigma''$ that Player $1$ can deviate to from $\sigma$.
\end{definition}

Informally, in an hSPE a player is winning from exactly these combinations of history and current vertex that she can win from at all. What differentiates the notion of hSPE from the usual subgame-perfect equilibria is that we take into account not just histories realizable in the arena, but also the \emph{hypothetical} histories which could not actually have happened. Clearly, every hSPE can be domain-restricted to an SPE. Moreover, if for a winning condition~$W$ there are SPE for all arenas, there are also hSPE for all arenas. The subtlety of this notion only becomes relevant when working with restricted classes of arenas.

We focus on the existence of \textit{finite-memory} (FM) hSPE for classes of games. An hSPE is finite-memory when the strategies $\sigma_i$, $i \in \{1, 2\}$ are finite-memory, i.e., can be implemented as Moore machines as seen above. The size of an FM hSPE is taken as the sum of the sizes of its strategies' Moore machines.

\subsection{State of the art}
\label{sec:sota}
We will briefly review a selection of results about combinations of classical objectives in two-player turn-based zero-sum games. As mentioned in Sect.~\ref{sec:intro}, results in this area are too numerous to provide an exhaustive list here. We thus focus on prominent winning conditions and memory requirements only, and do not delve too deep into technical details.

\smallskip\noindent\textbf{Simple objectives.} We first present the simple winning conditions.
\begin{itemize}
\item \textit{Parity.} Vertices are labeled with integer priorities, and a play is winning for Player~1 iff, among the priorities seen infinitely often, the maximal one is even. This condition subsumes many simple ones such as \textit{reachability} or \textit{safety} with regard to a given set.
\item \textit{Muller.} Let $U_1, \ldots, U_p$ be subsets of vertices, a play is winning iff the set of vertices seen infinitely often is equal to some $U_i$.
\item \textit{Energy.} Vertices are labeled with integer weights representing energy consumption and/or gain. The running sum of weights along a play must stay non-negative at all times. An upper bound may also be fixed. We consider two variants of upper bounds: \textit{battery-like} and \textit{spill-over-like} energy conditions. In the former, any energy in excess is simply lost, while in the latter, the play is lost if the upper bound is exceeded.
\item \textit{Mean-payoff.} Vertices are labeled with integer weights, and we look at the limit of the averages over all prefixes of a play. Since this limit need not exist in general, two variants are studied, for $\limsup$ and $\liminf$. We require the limit to be above a given threshold.
\item \textit{Total-payoff.} Vertices are labeled with integer weights, and we look at the limit of partial sums over prefixes. Again, two variants exist for $\limsup$ and for $\liminf$. We require the limit to be above a given threshold.
\item \textit{Average-energy.} Vertices are labeled with integer weights, and we look at the limit of the averages of partial sums (i.e., the average energy level) over prefixes. Again, two variants exist for $\limsup$ and for $\liminf$. We require the limit to be above a given threshold.
\item \textit{Window objectives.} Variations of mean-payoff and parity where instead of looking at the limit over an infinite play, we fix a finite window sliding over the play and we require the appropriate behavior to happen within the window at each step along the play.
\end{itemize}
Observe that all these winning conditions, along with arbitrary combinations, can be expressed in our formalism, i.e., as subsets of $C^\omega$ for an appropriate set of colors $C$.

\smallskip\noindent\textbf{Single-objective case.} The situation is as follows: parity~\cite{DBLP:journals/tcs/Zielonka98}, energy with only a lower bound~\cite{DBLP:conf/emsoft/ChakrabartiAHS03}, mean-payoff~\cite{EM79}, total-payoff~\cite{DBLP:conf/mfcs/GimbertZ04} and average-energy~\cite{DBLP:journals/acta/BouyerMRLL18} games are memoryless determined; Muller games require exponential-memory strategies for both players~\cite{DBLP:conf/lics/DziembowskiJW97}; lower and upper bounded energy games (both variants) require pseudo-polynomial memory (in the upper bound) for both players~\cite{DBLP:journals/acta/BouyerMRLL18}; window objectives require polynomial memory (in the window size) for both players~\cite{DBLP:journals/iandc/Chatterjee0RR15,DBLP:journals/corr/BruyereHR16}. Hence, \textit{all these winning conditions are FM determined}.

\smallskip\noindent\textbf{Combinations of objectives.} Despite this common trait, \textit{these objectives behave in very different ways when used in combinations}, even for very restricted ones. First, let us note that most results in the literature deal with the simplest case of conjunctions (or disjunctions) of objectives. In this setting, parity~\cite{DBLP:conf/fossacs/ChatterjeeHP07}, Muller, energy~\cite{DBLP:journals/acta/ChatterjeeRR14,DBLP:conf/icalp/JurdzinskiLS15}, and window objective~\cite{DBLP:journals/iandc/Chatterjee0RR15,DBLP:journals/corr/BruyereHR16} games remain FM determined (often, with an exponential blow-up in the number of conjuncts). The same holds for conjunctions of (possibly multiple) lower-bounded energy and parity conditions~\cite{DBLP:journals/tcs/ChatterjeeD12,DBLP:journals/acta/ChatterjeeRR14} or of a single average-energy condition with an energy one~\cite{DBLP:journals/acta/BouyerMRLL18,DBLP:conf/fossacs/BouyerHMR017}.

On the contrary, conjunctions of mean-payoff conditions~\cite{DBLP:journals/iandc/VelnerC0HRR15} or of a single mean-payoff and a single parity condition~\cite{DBLP:conf/lics/ChatterjeeHJ05} require infinite-memory strategies (for Player~1). Furthermore, it follows from the undecidability of multi-dimension average-energy games~\cite{DBLP:conf/fossacs/BouyerHMR017} and multi-dimension total-payoff games~\cite{DBLP:journals/iandc/Chatterjee0RR15} that unbounded memory is required for both players in these two settings. Finally, let us mention that \textit{Boolean combinations} were studied by Velner for mean-payoff conditions, and they were proved to be undecidable and requiring infinite-memory strategies~\cite{DBLP:conf/fossacs/Velner15}. For Boolean combinations of window objectives and a strict subset of parity conditions, Bruyère et al.~proved that games are exponential-memory determined~\cite{DBLP:conf/concur/BruyereHR16}.

The \textit{take-home message} is that winning conditions that are similar in the single-objective case may lead to contrasting behaviors when considering combinations, even for the simple case of conjunctions. Despite all these related works, little is known about the inherent mechanisms underlying FM determinacy, and why some objectives preserve it in combinations while others do not. Our main theorem, in Sect.~\ref{sec:thm}, formulates an answer to that question.

\subsection{A motivating example}
\label{sec:example}
We define a general class of games combining Muller conditions and instances of the two variants of bounded energy games presented in Sect.~\ref{sec:sota}.

\begin{definition}[Multi-dimension bounded-energy Muller games]\label{def:mdbemg}
Let $U_1,\dots,U_p$ be subsets of vertices of a finite arena, let every vertex of this arena be labeled with a tuple in $\mathbb{Z}^{n+m}$, let $b \in \mathbb{N}^{n+m}$, and let $\varphi$ be a proposition from propositional logic with $p+n+m$ free variables. Player~1 wins the game iff $\varphi(x_1,\dots,x_p,y_1,\dots,y_n,z_1,\dots,z_m)$ holds, where
\begin{itemize}
\item $x_i$ holds if the Muller condition induced by $U_i$ is satisfied,
\item $y_i$ holds if the the $i$-th battery-like energy condition holds, using $b_i$ as a battery-like bound and the $i$-th components of the labels as energy deltas,
\item $z_i$ holds if the the $i$-th spill-over-like energy condition holds, using $b_{i+n}$ as a spill-over-like bound and the $(i+n)$-th components of the labels as energy deltas.
\end{itemize}
\end{definition}

\smallskip\noindent\textit{Example.} A reporter has two options to travel across a country and make a documentary. Either by car and by sending pictures instantly to colleagues from her newspaper, or on foot to get a deeper understanding of the country. In the latter case she would not need to send pictures instantly, but to report on her location everyday for safety reasons. She finally decides to go by car, but if at some point she ran out of gas, she would continue on foot. In these specifications we can identify one energy condition, relating to the gas for the car, and two Muller conditions: sending the pictures and reporting the location. Note that such \textit{conditional Muller} conditions could not be simulated by one Muller condition alone.

Whereas all the ingredients of such a game, i.e., Muller, battery-like and spill-over-like energy conditions, are rather well understood, little is known about arbitrary combinations thereof. To prove that they are all FM determined we shall prove a more general result implying that combining a well-behaved winning condition (e.g., Muller) with conditions expressible by finite automata (e.g., bounded energy) yields a well-behaved condition again.

\section{Main theorem}
\label{sec:thm}

\subsection{The class of arenas}
We will formulate our main theorem for games built with arenas from some class $\mathcal{O}$ that is closed under certain operations. The class of all arenas is trivially closed under these operations, and in most applications it will be the most reasonable choice for $\mathcal{O}$. Our theorem allows for cases, however, where to ensure the premise we exploit some interactions between the specific structure of the arenas considered and the winning condition.
The two operations we consider are restrictions of the arena, and products with finite automata.

\begin{definition}[Restriction of an arena] Let $O = \langle V_1, V_2, E,\Gamma\rangle$ be an arena, and let $S\subseteq V_1 \cup V_2$ be such that $vE \cap S \neq \emptyset$ for all $v\in S$, where $vE := \{u \in V_1 \cup V_2\,\mid\, (v,u) \in E\}$. The arena $O\mid_S \coloneqq  \langle  V_1\cap S, V_2\cap S, E \cap S^2,\Gamma\mid_S\rangle$ is called the restriction of $O$ to $S$.
\end{definition}

\begin{definition}[Product arena-automaton]
Let $O = \langle V_1, V_2, E, \Gamma\rangle$ be an arena and let $\mathcal{A} = (C,Q, q_0,F,\Delta)$ be a finite automaton over the alphabet $C$ of colors. The product $O \times \mathcal{A}$ is the arena $\langle V_1 \times Q, V_2 \times Q, E',\Gamma'\rangle$ where $\Gamma'(v,q) \coloneqq \Gamma(v)$, and where $((v,q),(v',q'))\in E'$ iff $(v,v')\in E$ and $(q,\Gamma(v),q') \in \Delta$.
\end{definition}

\begin{lemma}
Let $O = \langle V_1, V_2, E, \Gamma\rangle$ be an arena and let $\mathcal{A} = (C,Q, q_0,F,\Delta)$ and $\mathcal{A}' = (C,Q', q'_0,F',\Delta')$ be finite automata.
\begin{enumerate}
\item (Idempotency) $(O\mid_S)\mid_{S'} = O\mid_{(S \cap S')}$
\item (Associativity) $(O \times \mathcal{A}) \times  \mathcal{A}' \cong O \times (\mathcal{A} \times  \mathcal{A}' )$.
\item (Restricted commutativity) $O\mid_S \times \mathcal{A} = (O \times \mathcal{A})\mid_{S \times Q}$
\end{enumerate}
\end{lemma}

\begin{corollary}
Let the arena $O'$ be obtainable from the arena $O$ by finitely many operations chosen from products with automata and restrictions. Then there is an automaton $\mathcal{A}$ and a suitable set $S \subseteq V \times Q$ such that $O' = (O \times \mathcal{A})\mid_S$, where $Q$ are the states of $\mathcal{A}$.
\end{corollary}
So if we start with some class of arenas, and then wish to close it under products with automata and restriction, we can just first consider all products of the original arenas with finite automata, and then the restrictions of these; and we obtain the desired class.

\subsection{Regularly-predictable games}
Our theorem requires, informally spoken, that Player 1 always knows whether winning is still possible for him given the history so far and given the current vertex. \emph{Knowing} here means (since we want FM strategies) the existence of a finite automaton producing the answer.

\begin{definition}
\label{def:predictable}
A game $g$ is \textit{regularly-predictable} if for all vertices $v\in V$ of $g$ there exists a finite automaton $\mathcal{A}_v$ that reads an initial color sequence and accept it iff Player $1$ has a winning strategy from $v$ after this sequence. A winning condition is regularly-predictable if all games using it are regularly-predictable.
\end{definition}

Many popular winning conditions are prefix-independent, in which case it does not depend on the history at all whether a player can win from some vertex. Hence, these are trivially regularly-predictable. Reachability and safety conditions, on the other hand, are not prefix-independent, but still regularly-predictable.

At first glance, one may think that FM determinacy implies regular-predictability. This is false, as witnessed by the following proposition, where $K_2$ is the two-clique with self-loops.

\begin{proposition}
Energy games with only a lower bound are not regularly-predictable.
\begin{proof}
Consider a one-player game on $K_2$, with one vertex giving $+1$ energy, and the other giving $-1$. The winning condition is that the current energy level stays non-negative. The player can win from every history where the energy level had not already been negative before. However, deciding this amounts to deciding the language of all binary words containing at least as many $1$s and $0$s, a typical example of a non-regular language. Hence such sequences cannot be recognized by a finite automaton.
\end{proof}
\end{proposition}

The proposition above together with the example below shows that regular-predictability and existence of FM hSPE are of incomparable strength.

\begin{example}
\label{example:irregular}
Let $W$ be the winning condition for Player $1$ that consists of the non-regular sequences in $\{0,1\}^\omega$: it is prefix-independent, so every game is regularly-predictable. Yet, the game over $K_2$ (with self-loops) involving colors $0$ and $1$ and where only Player $1$ plays is winnable but not via FM strategies (as it would contradict non-regularity of the sequences).
\end{example}

\subsection{Boolean combination}
The following definition suggests how we will combine the well-behaved regularly-predictable winning conditions to conditions definable by regular languages.

\begin{definition}[Regular combination of winning conditions]\label{def:rc-wc}
Let $\mathcal{W}$ be a class of subsets of $C^\omega$ that is closed under Boolean combination. The combination of $\mathcal{W}$ with $l$ regular languages, denoted $R_l(\mathcal{W})$, is defined as follows. Let $\varphi(w_1,\dots,w_k, r_1,\dots,r_l)$ be a propositional-logic formula with $k+l$ variables. Let $W_1,\dots,W_k \in\mathcal{W}$. For $1 \leq i \leq l$ let $L_i$ be a regular language over $C$. Then the winning condition $\{\rho\in C^\omega\mid\varphi(\rho \in W_1,\dots,\rho \in W_k, Pref(\rho) \cap L_1 = \emptyset, \dots, Pref(\rho) \cap L_l = \emptyset)\}$ is in $R_l(\mathcal{W})$.
\end{definition}

Intuitively, we allow for any Boolean combination between variables that either represent satisfaction of an objective from $\mathcal{W}$ (which we will require to be regularly-predictable) or represent the existence of a prefix of the play within a given regular language. Note that it is not restrictive for $\mathcal{W}$ to be closed under Boolean combination: we present our framework as such to allow for very general combinations (e.g., in the toy example from Section~\ref{sec:example}), but, for less amenable winning conditions, one can still take $\mathcal{W}$ to have four elements: a winning condition of interest, its complement, and the empty/universal winning conditions.

\subsection{Regular combinations preserve FM determinacy}

We are now able to state our main result. Due to the technicality of its proof, we first give a high-level sketch. The formal proof follows.

\begin{theorem}\label{thm:color-hspe}\hfill
\begin{itemize}
\item Let $\mathcal{O}$ be a class of arenas that is closed under product with finite automata and under simple restriction.
\item Let $\mathcal{W}$ be a class of winning conditions that is closed under Boolean combination.
\item Let all games in $\mathcal{O} \times \mathcal{W}$ be regularly-predictable and have FM hSPE.
 \end{itemize}
 Then all games in $\mathcal{O} \times R_l(\mathcal{W})$ are regularly-predictable and have FM hSPE. (for all $l \in \mathbb{N}$)
\end{theorem}

\begin{proof}[Proof sketch]
The proof of Theorem~\ref{thm:color-hspe} proceeds by induction on $l$ from Definition~\ref{def:rc-wc}. The basic idea of the induction step is to do the product of the arena of the game by a finite automaton accepting the language $L_l$, and then to partition the vertices of the new game into three regions: a region where $Pref(\rho) \cap L_l = \emptyset$ has been falsified, a region where the players can force the play\footnote{Slight abuse of notation as plays are formally defined as sequences of colors only: for the sake of readability, we sometimes use a play to refer to the corresponding sequence of vertices in the graph.} to reach the part of the first region that they like, and a third region restricted to the remaining vertices, where $Pref(\rho) \cap L_l = \emptyset$ holds, and no player has an incentive to leave the third region.

For the first and third regions, a suitable strategy profile will be provided by induction hypothesis, for the second region reachability analysis will do. These profiles together will be translated back into one single profile for the original game. A problem arises: a possible history in the original game may correspond in the product game to a history starting in the third region, going down the second region and coming back to the third. Thus the induced color sequence may fall out of the domain of the profile for the third region. Therefore we need to cope with ``impossible histories'', and use hSPE rather than the more familiar SPE.

A second problem arises because of the aforementioned reachability analysis. To show this, let us slightly detail the basic idea of the proof. Let $\mathcal{A}$ be a finite automaton recognizing the words with a prefix in $L_l$. The product $O \times \mathcal{A}$ looks like Fig.~\ref{fig:prod1}, where $T$ refers to the vertices where $Pref(\rho) \cap L_l = \emptyset$ has not (yet) been falsified. The bottom part corresponds to the product of $O$ with the final states of $\mathcal{A}$. Fig.~\ref{fig:prod2} invokes the induction hypothesis to get a suitable profile $\sigma_\bot$ for this part of the arena. Then in Fig.~\ref{fig:prod3} we would like to split $T$ into two regions: first, a region $T_1 \cup T_2$, where Player $i$ can force (as suggested by the arrow tips) every play starting in $T_i$ to reach the region she likes in the bottom component; and second, a region $T_0$ that no player wants to leave. The problem is that whether a player can win from a vertex in the bottom component depends on how the vertex was reached, so $T_1$ and $T_2$ cannot be simply defined as subsets of vertices in $O \times \mathcal{A}$. It is possible to reduce the problem to classical reachability, though: the following formal proof considers the product of the original arena $O$ with \textit{all} the automata provided by the regular-predictability assumption (as well as with other automata derived from $L_l$).
\end{proof}

\begin{figure}[thb]
\centering
\begin{minipage}[b]{1.5in}
\centering
\begin{tikzpicture}[shorten >=1pt,node distance=1.4cm, auto]
  \node[state] (x) {$T$};
  \node[state, rectangle] (y) [below of = x] {$Pref(\rho) \cap L_l \neq \emptyset$};

\path[->] (x) edge [bend left = 10] node {} (y)
		(x) edge [bend right = 10] node {} (y);
 \end{tikzpicture}
  \caption{The product arena $O \times \mathcal{A}$.}
  \label{fig:prod1}
 \end{minipage}
\hfill
 \begin{minipage}[b]{1.6in}
\centering
\begin{tikzpicture}[shorten >=1pt,node distance=1.4cm, auto]
  \node[state] (x) {$T$};
  \node[state] (y) [below of = x] {$\sigma_\bot$};

\path[->] (x) edge [bend left = 10] node {} (y)
		(x) edge [bend right = 10] node {} (y);
 \end{tikzpicture}
  \caption{Profile $\sigma_\bot$ provided by induction.}
  \label{fig:prod2}
 \end{minipage}
 \hfill
\begin{minipage}[b]{1.8in}
\centering
\begin{tikzpicture}[shorten >=1pt,node distance=1.4cm, auto]
  \node[state] (x) {$T_0$};
   \node[state] (x1)  [below left of = x]{$T_1$};
   \node[state] (x2)  [below right of = x]{$T_2$};
    \node[state] (y) [below right of = x1] {$\sigma_\bot$};

\path[-] (x) edge node {} (x1)
		(x) edge node {} (x2)
		(x1) edge node {} (x2)		;
\path[->] (x1) edge node {} (y)
		(x2) edge node {} (y);	
 \end{tikzpicture}
  \caption{Splitting $T$ via simplistic reachability analysis.}

 \label{fig:prod3}
\end{minipage}
 \end{figure}

For the formal proof, we need some additional notation: If $I$ is some finite index set, and $(\mathcal{A}_i)_{i \in I}$ is an $I$-indexed family of automata, we denote their product by $\otimes_{i\in I} \mathcal{A}_{i}$.

\begin{proof}[Proof of Theorem~\ref{thm:color-hspe}]
By induction on $l$. The claim holds for $l = 0$ by assumption since $R_0(\mathcal{W}) = \mathcal{W}$, so let us assume that $l > 0$.
Let $g = \langle V_1, V_2, E, \Gamma,W\rangle \in \mathcal{O} \times R_l(\mathcal{W})$, let $\varphi$, $W_1,\dots,W_k$, and $L_1,\dots,L_l$ witness that $W \in R_l(\mathcal{W})$.

Let us fix $\varphi_{\bot}(w_1,\dots,w_k,r_1,\dots,r_{l-1}) \coloneqq \varphi(w_1,\dots,w_k,r_1,\dots,r_{l-1}, \bot)$ , and let $g_{\bot} = \langle V_1, V_2, E, \Gamma,W_{\bot}\rangle$ where $\rho \in W_{\bot}$ if and only if $\varphi_{\bot}(\rho \in W_1,\dots,\rho \in W_k, Pref(\rho) \cap L_1 = \emptyset, \dots, Pref(\rho) \cap L_{l-1} = \emptyset)$. Since $W_{\bot} \in R_{l-1}(\mathcal{W})$, by induction hypothesis let $\sigma_{\bot}$ be an hSPE for $g_{\bot}$ and for all vertices $v\in V$ let $\mathcal{A}_{v}^{\bot} = \langle C,Q_v^{\bot},q_{0,v}^{\bot},F_v^{\bot},\Delta_v^{\bot}\rangle$ witness the regular-predictability of $g_{\bot}$.

In $g$, if and once the current color history and the current vertex have falsified $Pref(\rho) \cap L_l = \emptyset$, two issues are simplified: first, playing in $g$ according to $\sigma_{\bot}$ is optimal for both players; second, for each vertex $v \in V \coloneqq V_1 \cup V_2$, the automaton $\mathcal{A}_{v}^{\bot}$ accounts faithfully for who has a winning strategy when reaching $v$ after a given color history.

When $Pref(\rho) \cap L_l = \emptyset$ has not (yet) been falsified, it is more difficult for players both to know how to play optimally and to know who is the potential winner. A special case arises when a player can force the play to eventually falsify $Pref(\rho) \cap L_l = \emptyset$ while ensuring victory regardless of how falsification happens. Given arbitrary color history and vertex, two types of object are relevant for the players to know whether the special case obtains: for all $v \in V$, the $\mathcal{A}_{v}^{\bot}$ and a finite automaton $\mathcal{B}_v^l = (C,Q_v^l,q_{0,v}^l,F_v^l,\Delta_v^l)$ accepting the finite words $\alpha$ such that $\alpha\Gamma(v)$ has a prefix in $L_l$. (The $\mathcal{B}_v^l$ exist since $L_l$ is regular.) Let $g' \coloneqq \langle O \times \otimes_{u\in V}\mathcal{B}_u^l \times \otimes_{u\in V} \mathcal{A}_{u}^{\bot},W \rangle$, where $O$ is the arena of $g$. Note that to play using finite memory in $g'$ is like to play in $g$ using the same memory used in $g'$, plus some finite memory keeping track of the $\mathcal{B}_v^l$ and the $\mathcal{A}_{v}^{\bot}$.

Let $S_1$ ($S_2$) be the vertices of $g'$, i.e., of the form $(v, (q_u^l)_{u\in V}, (q_u^{\bot})_{u\in V})$ such that $q_v^l \in F_v^l$ and $q_v^{\bot} \in F_v^{\bot}$ ($q_v^{\bot} \notin F_v^{\bot}$). We observe, merely by rephrasing, that given a color history $\alpha$ and a vertex $v$, the special case obtains for Player $1$ ($2$) iff Player $1$ ($2$) has a winning strategy in $g'$ to reach $S_1$  ($S_2$) from $($v$, (q_u^l)_{u\in V}, (q_u^{\bot})_{u\in V})$ where the states $q_u^l$ and $q_u^\bot$ correspond to the states of $\mathcal{B}_u^l$ and the $\mathcal{A}_{u}^{\bot}$, respectively, after reading $\alpha$. Recognizing the special case in $g$ therefore amounts to reachability analysis in $g'$: let $S'_1$ ($S'_2$) be the vertices of $g' \setminus (S_1 \cup S_2)$ from where Player $1$ ($2$) can reach $S_1$ ($S_2$) regardless of how the opponent is playing. So Player $1$ ($2$) has a partial positional strategy in $g'$ that ensures uniformly that $S_1$ ($S_2$) is reached from $S'_1$ ($S'_2$). The combination of the two partial strategies translates back to $g$ into an FM partial strategy profile $\sigma'_\bot$ that uses only the  $\mathcal{B}_v^l$ and the $\mathcal{A}_{v}^{\bot}$ as auxiliary memory, and such that playing according to $\sigma_\bot \cup \sigma'_\bot$ is optimal for both players when the vertex in $g$ corresponds to a vertex in $S_1 \cup S_2 \cup S'_1 \cup S'_2$.

Furthermore, deciding who is the potential winner in $g$ at a vertex after some color history is easy when the corresponding vertex in $g'$ is in $S'_1 \cup S'_2$: if it is in $S'_1$, Player $1$ is the potential winner; otherwise it is Player $2$. So the automaton  $\otimes_{u\in V}\mathcal{B}_u^l \times \otimes_{u\in V} \mathcal{A}_{u}^{\bot}$ will help us witness regular-predictability of $g$ in the next paragraphs.

Let us now consider the remaining case, where $Pref(\rho) \cap L_l = \emptyset$ has not been falsified, and no player is able to force falsification to his own benefit. The corresponding color histories and vertices in $g$ translate in $g'$ to the set $V_{\top} \coloneqq (V\times \otimes_{u\in V} Q_u^l \times \otimes_{u\in V} Q_u^{\bot}) \setminus (S_1 \cup S_2 \cup S'_1 \cup S'_2)$. By definition of reachability, every vertex in $V_{\top}$ with an edge towards $S_1 \cup S_2 \cup S'_1 \cup S'_2$ has also an edge staying in $V_{\top}$. So $O'\mid_{V_{\top}}$ is also an arena, where $O'$ is the arena of $g'$. Moreover, by definition of $S'_1$, if Player $1$ controls a vertex in $V_{\top}$ with an edge towards $S_1 \cup S_2 \cup S'_1 \cup S'_2$, it is an edge towards $S_2 \cup S'_2$, so never taking this edge is optimal. (Likewise for Player $2$.) This implies that every (FM) hSPE in ``$g_{\top} \coloneqq g'\mid_{V_{\top}}$'' (precisely defined below) translates back to $g$ into (FM) partial hSPE, and for all color histories and vertices the potential winner in the restriction is the same as in $g'$. Let $\varphi_{\top}(w_1,\dots,w_k,r_1,\dots,r_{l-1}) \coloneqq \varphi(w_1,\dots,w_k,r_1,\dots,r_{l-1}, \top)$, and let $g_{\top} \coloneqq \langle O'\mid_{V_{\top}}, W_{\top}\rangle$ where $\rho \in W_{\top}$ iff $\varphi_{\top}(\rho \in W_1,\dots,\rho \in W_k, Pref(\rho) \cap L_1 = \emptyset, \dots, Pref(\rho) \cap L_{l-1} = \emptyset)$. Since $W_{\top} \in R_{l-1}(\mathcal{W})$, by induction hypothesis let $\sigma_{\top}$ be an hSPE for $g_{\top}$ and for all vertices $v\in V$ let $\mathcal{C}_{v}^{\top} = \langle C,Q_v^{\top},q_{0,v}^{\top},F_v^{\top},\Delta_v^{\top}\rangle$ witness the regular-predictability of $g_{\top}$.

To show the regular-predictability of $g$, let $v \in V$ and let $\mathcal{A}_v$ be the modification  of $\otimes_{u\in V} \mathcal{B}_u^l \times \otimes_{u\in V} \mathcal{A}_{u}^{\bot} \times  \mathcal{C}_{v}^{\top}$ where the final states $F_v$ are the $((q_u^l)_{u\in V},(q_u^{\bot})_{u\in V},q_v^{\top})$ such that either $(v,(q_u^l)_{u\in V},(q_u^{\bot})_{u\in V}) \in  S_1$ (falsification happened and Player $1$ is the potential winner), or $(v,(q_u^l)_{u\in V},(q_u^{\bot})_{u\in V}) \in  S'_1$ (Player $1$ can force falsification to his own benefit), or $(v,(q_u^l)_{u\in V},(q_u^{\bot})_{u\in V}) \notin S_2 \cup S'_2 \wedge q_v^{\top} \in F_v^{\top}$ (falsification benefiting Player $2$ has not happened and is not going to, and Player $1$ is the potential winner of restricted game). The automaton $\mathcal{A}_v$ accepts exactly the color histories after which Player $1$ has a winning strategy when starting from $v$. 

A strategy profile $\sigma$ for $g$ (which is meant to be an hSPE) is build by case disjunction. Informally, given a color history and a vertex,
\begin{itemize}
\item if the corresponding vertex in $g'$ is in $S_1 \cup S_2$, follow $\sigma_{\bot}$,
\item if it is in $S'_1 \cup S'_2$, follow $\sigma'_{\bot}$,
\item otherwise follow $\sigma_{\top}$.
\end{itemize}
More formally, let $\alpha\in C^*$, let $v \in V$, let $((q_u^l)_{u\in V},(q_u^{\bot})_{u\in V})$ be the state of $\otimes_{u\in V} \mathcal{B}_u^l \times \otimes_{u\in V} \mathcal{A}_{u}^{\bot}$ after reading $\alpha \Gamma(v)$, and let us make a case disjunction.
\begin{itemize}
\item  If $(v,(q_u^l)_{u\in V},(q_u^{\bot})_{u\in V}) \in S_1 \cup S_2$ let $\sigma(\alpha,v) \coloneqq \sigma_{\bot}(\alpha,v)$.
\item If $(v,(q_u^l)_{u\in V},(q_u^{\bot})_{u\in V}) \in S'_1 \cup S'_2$, let  $\sigma(\alpha,v) \coloneqq \pi_1\circ \sigma'_{\bot}(\alpha,v,(q_u^l)_{u\in V},(q_u^{\bot})_{u\in V})$
\item If $(v,(q_u^l)_{u\in V},(q_u^{\bot})_{u\in V}) \in V_{\top}$, let $\sigma(\alpha,v) \coloneqq \pi_1\circ \sigma_{\top}(\alpha,(v,(q_u^l)_{u\in V},(q_u^{\bot})_{u\in V}))$.
\end{itemize}
(Above, projection $\pi_1$ is used twice to retrieve vertex $v$ from vertex $(v,(q_u^l)_{u\in V},(q_u^{\bot})_{u\in V})$.) To implement the above strategy, we use the automaton $\otimes_{u\in V} \mathcal{B}_u^l \times \otimes_{u\in V} \mathcal{A}_{u}^{\bot}$ to keep track of the current vertex in $g'$, and we also use automata implementing $\sigma_{\bot}$, $\sigma'_{\bot}$, and $\sigma_{\top}$.
\end{proof}

\smallskip\noindent\textbf{Prefix-independence.} Many popular winning conditions are prefix-independent, in which case regular-predicta\-bi\-lity comes for free, and the notions of SPE and hSPE coincide. We state this special case explicitly, and also include an upper bound for the required memory.

\begin{corollary}\label{cor:color-spe}
Let $\mathcal{O}$ be a class of arenas that is closed under product with finite automata and under simple restriction. Let $\mathcal{W}$ be a class of prefix-independent winning conditions that is closed under Boolean combination. If all games in $\mathcal{O} \times \mathcal{W}$ have FM SPE, so do all games in $\mathcal{O} \times R_l(\mathcal{W})$.

Furthermore, let $f(l,n)$ be the memory size required to implement the optimal strategies for $n$-vertex games in $R_l(\mathcal{W})$ if the regular languages are recognized by automata with $m$ states each. Then $f(l,n) \leq m^n\cdot f(l-1,n) \cdot f(l-1,nm^{n})$.
\end{corollary}

\begin{proof}
Let us use the notation from the proof of Theorem~\ref{thm:color-hspe}. Thanks to prefix-independence of the winning condition, we do not need to worry about regular-predictability and the corresponding automata. Since $\sigma_{\bot}$ is an hSPE on a $n$-vertex game, $f(l-1,n)$ states suffice to implement it. The $\sigma'_{\bot}$ for reachability can be chosen memoryless. As $\sigma_{\top}$ is an hSPE on a game with at most $nm^{n}$ vertices, $f(l-1,nm^{n})$ states suffice to implement it. Additionally, we keep track of the current vertex in $g'$, adding a factor of $m^n$. Therefore $f(l,n) \leq m^n\cdot f(l-1,n) \cdot f(l-1,nm^{n})$.
\end{proof}

Note that while Corollary \ref{cor:color-spe} provides bounds on the memory requirements, these bounds are non-elementary. Such horrible bounds are unfortunately inevitable given the generality of the framework. In Sect.~\ref{sec:special}, we provide results for specific Boolean combinations that yield much more reasonable memory bounds.

Let us mention that very simple cases of combinations already lead to large lower bounds in terms of memory. For example, exponential memory is required for conjunctions of fully-bounded energy conditions, while they are actually all covered by the regular languages part of $R_l(\mathcal{W})$~\cite{DBLP:conf/fossacs/BouyerHMR017}. The same holds for conjunctions of window objectives~\cite{DBLP:journals/iandc/Chatterjee0RR15,DBLP:journals/corr/BruyereHR16}.

\smallskip\noindent\textbf{Back to the example.} We can now apply our theorem to the example of Sect.~\ref{sec:example}.

\begin{corollary}\label{cor:md-be-mg-fmd}
Multi-dimension bounded-energy Muller games have FM SPE (aka combination of optimal strategies).
\begin{proof}
Let $\mathcal{O}$ be the class of the finite arenas, which is clearly closed under simple restriction and product by finite automaton. Let $\mathcal{W}$ be the class of the (prefix-independent) Muller winning conditions, which is closed under Boolean combination. All games in $\mathcal{O} \times \mathcal{W}$ have FM SPE~\cite{DBLP:conf/lics/DziembowskiJW97}, and so do all games in $\mathcal{O} \times R_l(\mathcal{W})$ by Corollary~\ref{cor:color-spe}.

To show that this applies to multi-dimension bounded-energy Muller games, it suffices to see that fully-bounded energy conditions can be expressed as regular languages, which is trivial since the sum of weights must be constrained within the two bounds at all times, hence can only take a bounded number of integer values.
\end{proof}
\end{corollary}

\section{Applications}
\label{sec:app}

Let us compare the hypotheses and concepts used in Theorem~\ref{thm:color-hspe} with the classical (combinations of) winning conditions from the literature, already discussed in Sect.~\ref{sec:sota}, to get a better grasp of its applicability.
Interesting winning conditions for our theorem fall in two categories: regular languages and regularly-predictable conditions. First, observe that any condition that can be defined as a regular language is trivially regularly-predictable. (The latter is a more general notion.) Hence, we first discuss the regular languages.

\smallskip\noindent\textbf{Regular languages.} Among the simplest conditions that can be recognized through finite automata (hence expressed as regular languages) lie \textit{reachability} and \textit{safety} conditions. Indeed, as mentioned in Sect.~\ref{sec:intro}, regular languages are essentially a compact way to represent \textit{safety-like} winning conditions. For example, fully-bounded energy conditions (both battery-like and spill-over-like) can also be encoded through finite automata whose size depends on the upper bound. Window objectives (both for mean-payoff and parity), thanks to their finite window mechanism, can also be represented as regular languages. All the other objectives discussed in Sect.~\ref{sec:sota} cannot.

\smallskip\noindent\textbf{Regularly-predictable conditions.} Recall that for Theorem~\ref{thm:color-hspe}, we require the class $\mathcal{W}$ to be closed under Boolean combination and such that winning conditions are regularly-predictable and admit FM hSPE. Let us review the classical objectives.
\begin{itemize}
\item \textit{Regular languages.} As stated above, conditions expressed as regular languages can be used. Hence, this easily permits to rediscover FM determinacy results for multi-dimension fully-bounded energy games~\cite{bouyer,DBLP:journals/acta/BouyerMRLL18,DBLP:conf/fossacs/BouyerHMR017} or conjunctions of window objectives~\cite{DBLP:journals/iandc/Chatterjee0RR15,DBLP:journals/corr/BruyereHR16}, and even extend them to full Boolean combinations.
\item \textit{Parity and Muller.} Any combination of such conditions can be expressed in the closed class. Furthermore they are trivially regularly-predictable (because they are prefix-independent) and admit FM hSPE. Hence, these conditions can be mixed in any Boolean combination with regular languages and retain FM determinacy. This lets us rediscover FM determinacy results for generalized parity games~\cite{DBLP:conf/fossacs/ChatterjeeHP07}, or combinations of parity conditions with window conditions~\cite{DBLP:conf/concur/BruyereHR16}, and extend them to full Boolean combinations.
\item \textit{Mean-payoff.} The mean-payoff condition is regularly-predictable (as it is prefix-indepen\-dent) and admits FM hSPE. Unfortunately, this does not hold for Boolean combinations of mean-payoff~\cite{DBLP:journals/iandc/VelnerC0HRR15,DBLP:conf/fossacs/Velner15}. Still, one can take $\mathcal{W}$ as the trivial class containing one mean-payoff condition and its complement, and use it in Boolean combinations with regular languages.
\item \textit{Average-energy, total-payoff and energy with no upper bound.} These three conditions are not regularly-predictable as one needs to be able to store an arbitrarily large sum of weights in memory to decide if Player~1 can win from a given prefix. Hence our theorem cannot be applied to these conditions.
\end{itemize}

\smallskip\noindent\textbf{Theorem applicability.} Observe that our theorem cannot be applied to Boolean combinations of mean-payoff, average-energy and total-payoff, or to combinations of mean-payoff and parity (two regularly-predictable conditions but which cannot be put in the same closed class $\mathcal{W}$), four cases in which \textit{indeed FM determinacy is not preserved}~\cite{DBLP:journals/iandc/VelnerC0HRR15,DBLP:conf/fossacs/Velner15,DBLP:journals/acta/BouyerMRLL18,DBLP:journals/iandc/Chatterjee0RR15,DBLP:conf/lics/ChatterjeeHJ05}. On the other hand, we recover many known results from the literature~\cite{bouyer,DBLP:journals/acta/BouyerMRLL18,DBLP:conf/fossacs/BouyerHMR017,DBLP:journals/iandc/Chatterjee0RR15,DBLP:journals/corr/BruyereHR16,DBLP:conf/fossacs/ChatterjeeHP07,DBLP:conf/concur/BruyereHR16} and are able to extend them to more general combinations (or to completely novel ones).

It remains to discuss \textit{corner cases}: combinations that are known to preserve FM determinacy but are not covered by our theorem. We are aware of three cases from the literature, all involving the energy condition without an upper bound: (a)~conjunctions of energy conditions~\cite{DBLP:journals/acta/ChatterjeeRR14,DBLP:conf/icalp/JurdzinskiLS15}, (b)~conjunctions of energy and parity conditions~\cite{DBLP:journals/tcs/ChatterjeeD12,DBLP:journals/acta/ChatterjeeRR14}, (c)~conjunctions of energy and a single average-energy condition~\cite{DBLP:conf/fossacs/BouyerHMR017}. It is very interesting to spend a moment on these corner cases. Indeed, the ad-hoc techniques used to prove FM determinacy in all three cases intuitively rely on proving equivalence with games where the energy condition can be bounded \textit{both} from below \textit{and from above}, for a sufficiently large bound. Yet, we know that such fully-bounded energy conditions define regular languages. Hence, for cases (a) and (b) we actually retrieve applicability of our theorem, leaving case~(c) as the only case, to our knowledge, of preservation of FM determinacy in the literature which is not covered by Theorem~\ref{thm:color-hspe}. This is because the average-energy condition is not regularly-predictable (as one can see in~\cite{DBLP:journals/acta/BouyerMRLL18,DBLP:conf/fossacs/BouyerHMR017}, it behaves rather oddly in comparison to all other classical objectives).

\smallskip\noindent\textbf{FM determinacy vs.~FM hSPE}. We can say more on case~(c), and for that we want to highlight once again that the result we obtain in Theorem~\ref{thm:color-hspe} deals with a stronger concept than FM determinacy, namely the existence of FM (h)SPE. As seen above, these two notions do coincide in virtually all cases studied in the literature. Now, case~(c) is actually the only setting, to our knowledge, where they do not, as proved in the following example.

\begin{example}
Consider the arena in Fig.~\ref{fig:AEL}, colored by integers: vertex $a$ has weight $1$, $b$ weight $-1$ and $c$ weight $0$. The objective of Player $1$ (circle) is to have the average-energy (AE, limit of the average energy level) less than or equal to zero while keeping the energy level (EL) non-negative at all time. Such games are FM determined~\cite{DBLP:conf/fossacs/BouyerHMR017}, and in this case, Player $2$ (diamond) has a trivial strategy to win from $a$: looping forever. He also wins from $b$ (the energy directly drops below zero), and Player $1$ wins from $c$.

Now, consider SPE. Observe that when in vertex $c$, the EL does not change anymore so the AE is actually equal to the EL when you reach $c$. Thus, Player $1$ can win if he reaches $c$ with an EL of zero (and it is the only way). In an SPE, we need to consider all possible histories. Imagine an history $w = a^n$ for some $n \geq 1$: Player 1 wins by looping exactly $n$ times in $b$ before reaching $c$. This defines a subgame-perfect strategy for Player 1 which consists in looping in $b$ for as many times as Player $2$ looped in $a$. Clearly, this SPE requires infinite memory. We thus have an FM determined game where no FM SPE exists.

\begin{figure}[tb]
\begin{center}
\begin{tikzpicture}[node distance=1cm, auto]
  \node (R) {$v_0$};
  \node[draw,diamond] (A) [right=of R] {$a$};
  \node[draw,circle] (B) [right=2cm of A] {$b$};
  \node[draw,circle] (C) [right=2cm of B] {$c$};
  \node (w1) at (22mm,3mm) {$1$};
  \node (w2) at (50mm,3mm) {$-1$};
  \node (w3) at (74mm,3mm) {$0$};
  \draw[->] (R) to node {} (A);
  \draw[->] (A) to node {} (B);
  \draw[->] (B) to node {} (C);
   \draw[->] (C) edge [loop right] node {} ();
   \draw[->] (A) edge [loop above] node {} ();
   \draw[->] (B) edge [loop above] node {} ();
\end{tikzpicture}
\end{center}
\vspace{-5mm}
\caption{Average-energy games with lower-bounded energy are FM determined but do not admit FM SPE. Vertices are colored by integer weights.}
\label{fig:AEL}
\end{figure}
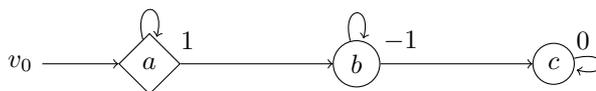
\end{example}

\section{Discussion of the requirements}
\label{sec:req}
We discuss the requirements of Theorem~\ref{thm:color-hspe} and explore whether improvement seems plausible.

\smallskip\noindent\textbf{Arenas.} First, we inspect the conditions pertaining to the arenas. Typically all arenas are considered, and the class of all arenas trivially satisfies the closure properties we demand. These properties are only relevant if we need to use specific subsets of the arenas to prove the required properties of the games. To ask that the class of arenas we consider is closed under products with finite automata is a very mild condition. Apart from arenas of bounded size, we would expect all naturally occurring classes of arenas to have this property.

Requiring closure under simple restriction is more restrictive (e.g.,~the games we consider in the upcoming Example \ref{ex:subgamenotimplied} are not closed under simple restriction). To see that this is necessary, let $O'$ be a simple restriction of an arena $O$ such that the game played on $O$ is FM determined, but the one played on $O'$ is not. W.l.o.g., assume that $O'$ has no deadlock. We can obtain a new winning condition for a game played on $O$ by combining the original condition with a regular language expressing informally that, if the game ever reaches a vertex in $O \setminus O'$, the last player to move loses, and otherwise, the winner is the winner of the original winning condition. The resulting game is not FM determined, thus witnessing the need for our requirement of closure under simple restriction.

Furthermore, as soon as we are demanding that our arenas are closed under products with automata and simple restriction, we see that we need to require the existence of FM SPE rather than merely FM determinacy (up to uniformity). The reason is that by taking a suitable combination of products and restriction, we generate an arena that first produces a predetermined finite color sequence before starting the original game.

\smallskip\noindent\textbf{Regular languages.} Theorem~\ref{thm:color-hspe} is tight in the sense that it fails if making the conjunction of the universal winning condition and a condition derived from any irregular language. Indeed, Example~\ref{example:irregular} already showed that FM hSPE may not exist for such conditions.

\smallskip\noindent\textbf{FM determinacy vs.~FM hSPE.} Being able to win using bounded finite memory, and being able to decide who wins from a given history with a finite automaton (i.e., regular-predictability) together do not suffice to imply the existence of an FM subgame-perfect strategy, as observed in the next example.

\begin{example}
\label{ex:subgamenotimplied}
Consider an arena with colors $0, 1, \alpha, \beta$, where a vertex has an $\alpha$-successor iff it has a $\beta$-successor, and each such vertex is controlled by Player $1$ (let us call them $\alpha\beta$-vertices). Player $1$ wins any play $\rho$ that has a prefix of the form $w\alpha$ with $w \in \{0,1\}^*$ and $w$ has the same number of $0$ and $1$s, and any play $\rho$ that has a prefix of the form $w\beta$ with $w \in \{0,1\}^*$ and $w$ has different numbers of $0$ and $1$s. Hence, for Player 1 to win, it suffices to reach an $\alpha\beta$-vertex and pick the $\alpha$ (resp.~$\beta$) successor if the play contains an equal (resp.~a different) number of $0$ and $1$s.
So there is a simple finite automaton deciding from which histories Player $1$ can win, and linear memory suffices for Player 1 to win from any winnable history, but Player $1$ has no FM subgame-perfect strategy.
\end{example}

\begin{proof}
Player $1$ wins from some not-yet-determined history iff he can force an $\alpha\beta$-vertex. Linear memory suffices to keep track of how many $0$ and $1$'s have been encountered along the way (as a simple path suffices), and enables Player $1$ to choose correctly. However, which choice is correct depends on the pre-history in a way that a finite automaton cannot keep track of (as this history is unbounded). Thus, there is no FM subgame-perfect strategy.
\end{proof}

In Sect.~\ref{sec:app}, we also discussed the case of conjunctions of energy and a single average-energy condition~\cite{DBLP:conf/fossacs/BouyerHMR017}, which is also FM determined, but do not have FM hSPE, for similar reasons.

Note that the games  above are not closed by simple restriction. As a further example, we show that the condition on subgame-perfect strategies is not dispensable entirely.

\begin{example}
Consider the arena in Fig.~\ref{fig:SPE} where Player $1$ controls circle vertices and Player~$2$ controls diamond vertices.
\begin{figure}[tb]
\begin{center}
\begin{tikzpicture}[node distance=1cm, auto]
  \node (R) {$v_0$};
  \node[draw,circle] (A) [below=2mm of R] {$a$};
  \node[draw,diamond] (B) [right=2cm of A] {$b$};
  \node[draw,diamond] (C) [right=2cm of B] {$c$};
  \node[draw,circle] (D) [left=of A] {$d$};
  \node[draw,circle] (E) [left=2cm of D] {$e$};
  \draw[->] (R) to node {} (A);
  \draw[->] (A) to node {} (D);
  \draw[->] (A) to[out=30, in=150] node {} (B);
  \draw[->] (B) to[out=-150, in=-30] node {} (A);
  \draw[->] (B) to node {} (C);
  \draw[->] (D) to[out=-150, in=-30] node {} (E);
  \draw[->] (E) to[out=30, in=150] node {} (D);
   \draw[->] (C) edge [loop right] node {} ();
     \draw[->] (E) edge [loop left] node {} ();
\end{tikzpicture}
\end{center}
\vspace{-8mm}
\caption{FM SPEs are needed for transfer in combinations, not only FM determinacy.}
\label{fig:SPE}
\end{figure}
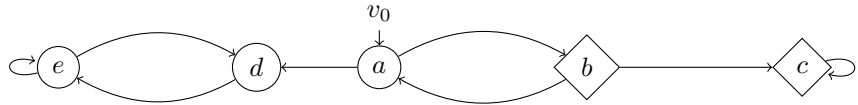
Let $W$ be defined such that Player $1$ wins iff either the play ends with $c^\omega$, or $b$ occurs an even number of times and the play ends with $deedeeedeeee\ldots$, or $b$ occurs an odd number of times and the play ends with $e^\omega$. Player $1$ has an FM winning strategy: play towards $b$ first; if the game returns to $a$, play $d$ and then loop at $e$. Player $1$ also has a subgame-perfect strategy, but no FM subgame-perfect strategy (as it needs to react to all histories, hence also the ones where $b$ happens an even number of times).

If we take the conjunction of $W$ with avoiding the regular language of words containing $c$, then in the resulting game, Player $1$ still wins, but he has no FM winning strategy.
\end{example}

\section{Specific results for conjunctions and disjunctions}
\label{sec:conj}
\label{sec:special}
The ability of our main theorem to deal with arbitrary Boolean formulae is often not needed. In fact, as discussed in Sect.~\ref{sec:sota}, in the literature on multi-dimension games, it is very common to work only with conjunctions on the winning conditions of Player 1 (which dually means using disjunctions for Player 2). In this section, we give some direct constructions for these cases, which have different requirements or conclusions from our main theorem -- and in particular, come with much more reasonable bounds than provided by Corollary \ref{cor:color-spe}.

As an outline, we give a brief account of the main constructions here. (i)~For simple disjunctions for Player 2, we obtain FM determinacy results without the regularity assumption on the combined language-based condition. (ii)~For simple conjunctions for Player 1, we obtain FM determinacy results and better memory bounds without requiring regular-predictability (but requiring regularity of the language). As corollary, we regain known bounds on fully-bounded energy games. (iii)~For the case of subgame-perfect strategies (not required in (i) and (ii)), we obtain better memory bounds. The interest of these side results, apart from improved bounds for still general classes of games, is to illustrate how the different hypotheses of Theorem~\ref{thm:color-hspe} interact and how restrictions on some dimensions of the problem permits to be more general on other dimensions.

The reason why we cannot obtain a general statement as in our main theorem by combining the lemmata in this section is that their conclusions do not match their requirements. We can thus not apply them in an iterative fashion.

To formulate some of the results, we recall the notion of \emph{future game} from~\cite{DBLP:journals/iandc/RouxP18}.
\begin{definition}[Future games]
Let $g$ be a game with winning condition $W$, and let $w \in C^*$. The future game $g^w$ is derived from $g$ by replacing $W$ with $\{ \rho\in C^\omega\,\mid\, w \cdot  \rho\in W\}$. If $\sigma_i$ is a strategy in $g$, let $\sigma_i^w(w',v) \coloneqq \sigma_i(ww',v)$ define another strategy for the same player.
\end{definition}

Our first lemma consider the case of disjunctions for Player~2, dropping the regularity assumption with regard to the language-based condition.

\begin{lemma}\label{lem:disj1}
Consider a game such that if Player $2$ can win $g$ or its future games, he can win by using finite memory. Let $L$ be a language over the alphabet $C$, and let us derive $g_{L}$ from $g$ by replacing $W$ with $W_{L}$ such that $\rho \in W_{L}$ iff $\rho \in W\,\vee\, Pref(\rho ) \cap L = \emptyset$. If Player $2$ wins $g_{L}$, he has an FM winning strategy.
\begin{proof}
Among the histories $h$ in $g$ whose colors are in $L$, let $\mathcal{H}_{L}$ contain the minimal ones for the prefix relation. Let $\mathcal{H}_2$ be the histories $h$ in $\mathcal{H}_{L}$ from where Player $2$ wins the future game of $g$ (and therefore also of $g_{L}$) after $h$. For a strategy $\sigma$, let $\mathcal{H}(\sigma)$ denote the set of histories compatible with $\sigma$, and let $[\mathcal{H}(\sigma)]$ denote the set of plays compatible with $\sigma$.

Let $\sigma_2$ be a winning strategy for Player $2$, so every play in $[\mathcal{H}(\sigma_2)]$ has a prefix in $\mathcal{H}_2$, and we can define $T$ as the subtree (a subset that is a tree) of $\mathcal{H}(\sigma_2)$ whose maximal paths are all in $\mathcal{H}_2$. By K\"onig's Lemma $T$ is finite. For each $h \in \mathcal{H}_2 \cap T$ let $\sigma_2^h$ be an FM strategy making Player $2$ win the future game of $g$ after $h$. The following FM strategy makes Player $2$ win $g_{L}$: if $h\in T \setminus \mathcal{H}_2$ go to some vertex $v \in V$ such that $hv \in T$; if $h$ has $h_1 \in \mathcal{H}_2$ as a prefix, play as prescribed by $\sigma_2^{h_1}$.
\end{proof}
\end{lemma}

We now turn to conjunctions for Player $1$, assuming regularity of the language, but dropping the regular-predictability assumption for the winning condition. For every finite automaton $\mathcal{A}$ let $L_{\mathcal{A}}$ be the words accepted by $\mathcal{A}$.

\begin{lemma}\label{lem:conj1}
Let $g = (O, W)$ be a game such that whenever $O'$ is a simple restriction of a product $O \times \mathcal{A}$ for some finite automaton $\mathcal{A}$, then $(O',W)$ is determined \textit{via} strategies using finite memory (of size $m(n)$, where $n$ is the number of vertices in $O'$). Let $\mathcal{A}$ be a finite automaton over the alphabet $C$, and let us derive $g_{\mathcal{A}}$ from $g$ by replacing $W$ with $W_{\mathcal{A}}$ such that $\rho \in W_{\mathcal{A}}$ iff $\rho \in W \wedge Pref(\rho ) \cap L_{\mathcal{A}} = \emptyset$. Then $g_{\mathcal{A}}$ is determined via strategies using finite memory (of size $|Q|\cdot m(|V|\cdot |Q|)$, where $Q$ are the states of $\mathcal{A}$).

\begin{proof}
Let $\mathcal{A} = (E,Q,q^0,F,\Delta)$ and let $g' \coloneqq ((O \times \mathcal{A}),W)$. Let $S \subseteq V\times Q$ be the vertices of $g'$ from where Player $2$ cannot force the play to reach $V \times F$. Let us make a case distinction. First case, $(v_0,q^0)$ is not in $S$. So in $g'$ Player $2$ can win by playing positionally. Such a strategy yields a strategy in $g$ ($g_{\mathcal{A}}$) that only needs memory to run $\mathcal{A}$, in order to know the current state in $Q$. For this a memory of size $|Q|$ suffices.

Second case, $(v_0,q^0) \in S$. By assumption, the simple restriction $(S,W)$ is determined via strategies using finite memory (of size $m(|V|\cdot |Q|)$). If Player $1$ has a winning strategy for $(S,W)$, he can use it together with $\mathcal{A}$ to play in $g_{\mathcal{A}}$.
If Player $2$ has a winning strategy, he can do the same until the play, seen as a play in $g'$, leaves $S$; and then he can play as in the first case.
\end{proof}
\end{lemma}

The next example shows that the regularity assumption of $L_{\mathcal{A}}$ in Lemma~\ref{lem:conj1} is not dispensable. We will use the language made of the histories with positive energy levels at every prefix, which is not regular.

\begin{example}
\label{ex:energynondetermined}
We consider games where the vertices are colored by $\{a,b\}$ and where Player $1$ wins the plays with colors $wa^\omega$ for all $w \in \{a,b\}^*$ and the plays with colors $babbaa\ldots b^na^n \ldots$. Such games are FM determined, but their conjunction with a lower-bounded energy condition is not.
\begin{proof}
In order to win, Player $1$ needs to be able to force arbitrarily long color sequences of $a$'s. But if he can do that, he can force eventually constant $a$ even by a positional strategy. Likewise, for Player $2$ it suffices to play a positional strategy preventing eventually constant $a$ to win all games he can win. Hence, the base game is memoryless determined.

To see that the energy version is not FM determined, we consider a one-player game on the 2-clique $K_2$ (with added self-loops), such that both states $s$ and $t$ belong to Player~1. Let $s$, the initial vertex, have color $b$ and weight $1$, and $t$ have color $a$ and weight $-1$. Player $1$ can win by playing according to the colors $babbaabbb\ldots$, which keeps the energy non-negative, but requires infinite memory to do so. However, any FM strategy winning the underlying game has to produce a color sequence of the form $wa^\omega$, which will cause the energy to diverge to $-\infty$, thus falsifying the energy condition.
\end{proof}
\end{example}

Next, we consider simple disjunctions for Player 1 and obtain improved bounds.

\begin{lemma}\label{lem:disjunction-player1}
Let $g = (O, W)$ be a regularly-predictable game such that whenever $O'$ is a simple restriction of a product $O \times \mathcal{A}$ for some finite automaton $\mathcal{A}$, then $(O',W)$ is determined, and whenever Player $1$ can win, he can do so \textit{via} strategies using finite memory (of size $m(n)$, where $n$ is the number of vertices in $O'$).

Let $\mathcal{A}$ be a finite automaton over the alphabet $C$, and let us derive $g_{\mathcal{A}}$ from $g$ by replacing $W$ with $W_{\mathcal{A}}$ such that $\rho \in W_{\mathcal{A}}$ iff $\rho \in W \vee Pref(\rho)\cap L_{\mathcal{A}} = \emptyset$. If Player $1$ has a winning strategy in $g_{\mathcal{A}}$, he has one using finite memory (of size $|Q| \cdot |Q_g| \cdot m(|V|\cdot|Q|\cdot|Q_g|)$), where $Q$ and $Q_g$ are the states of $\mathcal{A}$ and of a finite automaton witnessing regular-predictability of $g$.
\begin{proof}
Let a finite automaton $\mathcal{A}_g$ accept exactly the histories that correspond to the future games won by Player $2$. Let $\mathcal{A} = (E,Q,q^0,F,\Delta)$, let $\mathcal{A}_g = (E,Q_g,q^0_g,F_g,\Delta_g)$, and let $g' \coloneqq g \times \mathcal{A} \times \mathcal{A}_g$. Let $S \subseteq V\times Q \times Q_g$ be the vertices of $g'$ from where Player $2$ cannot force the play to reach $V\times F \times F_g$. Let us make a case distinction.

First case, $(v_0,q^0,q^0_g)$ is not in $S$, so in $g'$ Player $2$ can reach $V\times F \times F_g$ and win, by playing in some way.

Second case, $(v_0,q^0,q^0_g) \in S$. If a finite history $h$ in $g'\mid_S$ reaches $V\times F \times Q_g$, it must be in $V\times F \times (Q_g\setminus F_g)$, by definition of $S$. So Player $1$ can win after $h$ in $g'$, and therefore in the simple restriction $g'\mid_S$ too. By assumption he can do so \textit{via} a strategy using finite memory (of size $m(|V|\cdot |Q| \cdot |Q_g|)$). He can use the same strategy to win $g_{\mathcal{A}}$, but he needs more memory to run $\mathcal{A}$ and $\mathcal{A}_g$ in parallel.
\end{proof}
\end{lemma}

Since the assumptions from Lemma~\ref{lem:disjunction-player1} are stronger than those from Lemma~\ref{lem:disj1}, we also find that Player $2$ has an FM winning strategy if he can win. However, even under the assumptions from Lemma~\ref{lem:disjunction-player1}, the memory required by Player $2$ might not be uniformly bounded. We can obtain bounds for the memory required by Player $2$ (Lemma~\ref{lem:disj2}) by considering subgame-perfect strategies.

\begin{lemma}
\label{lem:constructsubgameperfect}
Let Player $1$ be able to win $g$ and all its future games that he wins by some FM strategy (using memory of size $k$). Moreover, for each FM strategy of size $k$, let there be a finite automaton (of size $l$) accepting exactly those histories won by that strategy. Then Player $1$ has an FM subgame-perfect strategy of size $(2l \cdot k)^{((k|V|)^{k|V|})}$.
\begin{proof}
There are $(k|V|)^{k|V|}$ FM strategies of size $k$. For each of them we keep track of its current state, and run the automaton deciding whether it is currently winning (using $l$ bits). We always play according to some strategy, and we change it only once it is no longer winning, and another strategy is identified as winning from the current history.
\end{proof}
\end{lemma}

\begin{lemma}\label{lem:disj2}
Let $g$ be a regularly-predictable game such that Player $2$ has a subgame-perfect strategy of size $m$. Let $\mathcal{A}$ be a finite automaton over the alphabet $C$. Let us derive $g_{\mathcal{A}}$ from $g$ by replacing $W$ with $W_{\mathcal{A}}$ such that $\rho \in W_{\mathcal{A}}$ iff $\rho \in W\,\vee\,Pref(\rho ) \cap L_\mathcal{A} = \emptyset$. If Player $2$ wins $g_{\mathcal{A}}$, he has a winning strategy of size $l|\mathcal{A}|+m$, where some automaton of size $l$ witnesses regular-predictability of $g$.
\begin{proof}
We can combine the automaton deciding who wins a history and $\mathcal{A}$ into one automaton of size $l|\mathcal{A}|$that accepts the intersection of these two languages. By considering the expansion by this automaton, we see that simulating it suffices for Player $2$ to force an accepted history, if he can do so. Now reaching such a history and then switching to the subgame-perfect FM strategy of size $m$ wins $g_\mathcal{A}$ if this is possible at all.
\end{proof}
\end{lemma}

We briefly return to the two versions of fully-bounded energy conditions from~\cite{bouyer} discussed in Sect.~\ref{sec:sota}: battery-like and spill-over-like variants. Let $E_{\max}$ be the energy upper bound in both cases. Since both conditions can be simulated by a finite automaton with $\mathcal{O}(E_{\max})$ states, Lemma~\ref{lem:conj1} implies the following corollary. Contrast this to the requirement of memory size $\vert V\vert \cdot d \cdot W$ (where $d$ is the number of priorities and $W$ the largest energy weight) for unbounded-energy parity games from~\cite{DBLP:journals/tcs/ChatterjeeD12}.

\begin{corollary}
\label{corr:batterylike}
Battery-like energy parity games and spill-over-like energy parity games are determined \textit{via} strategies using $\mathcal{O}(E_{\max})$ memory states, where $E_{\max}$ is the energy upper bound.
\end{corollary}

\bibliography{references}

\begin{thebibliography}{10}

\bibitem{DBLP:reference/mc/BloemCJ18}
Roderick Bloem, Krishnendu Chatterjee, and Barbara Jobstmann.
\newblock Graph games and reactive synthesis.
\newblock In Edmund~M. Clarke, Thomas~A. Henzinger, Helmut Veith, and Roderick
  Bloem, editors, {\em Handbook of Model Checking.}, pages 921--962. Springer,
  2018.
\newblock \href {http://dx.doi.org/10.1007/978-3-319-10575-8\_27}
  {\path{doi:10.1007/978-3-319-10575-8\_27}}.

\bibitem{bouyer}
Patricia Bouyer, Uli Fahrenberg, Kim~G. Larsen, Nicolas Markey, and Ji\v{r}\'i
  Srba.
\newblock Infinite runs in weighted timed automata with energy constraints.
\newblock In Franck Cassez and Claude Jard, editors, {\em Formal Modeling and
  Analysis of Timed Systems}, volume 5215 of {\em Lecture Notes in Computer
  Science}, pages 33--47. Springer Berlin Heidelberg, 2008.
\newblock \href {http://dx.doi.org/10.1007/978-3-540-85778-5\_4}
  {\path{doi:10.1007/978-3-540-85778-5\_4}}.

\bibitem{DBLP:conf/fossacs/BouyerHMR017}
Patricia Bouyer, Piotr Hofman, Nicolas Markey, Mickael Randour, and Martin
  Zimmermann.
\newblock Bounding average-energy games.
\newblock In Javier Esparza and Andrzej~S. Murawski, editors, {\em Foundations
  of Software Science and Computation Structures - 20th International
  Conference, {FOSSACS} 2017, Held as Part of the European Joint Conferences on
  Theory and Practice of Software, {ETAPS} 2017, Uppsala, Sweden, April 22-29,
  2017, Proceedings}, volume 10203 of {\em Lecture Notes in Computer Science},
  pages 179--195, 2017.
\newblock \href {http://dx.doi.org/10.1007/978-3-662-54458-7\_11}
  {\path{doi:10.1007/978-3-662-54458-7\_11}}.

\bibitem{DBLP:journals/acta/BouyerMRLL18}
Patricia Bouyer, Nicolas Markey, Mickael Randour, Kim~G. Larsen, and Simon
  Laursen.
\newblock Average-energy games.
\newblock {\em Acta Inf.}, 55(2):91--127, 2018.
\newblock \href {http://dx.doi.org/10.1007/s00236-016-0274-1}
  {\path{doi:10.1007/s00236-016-0274-1}}.

\bibitem{DBLP:conf/lata/BrenguierCHPRRS16}
Romain Brenguier, Lorenzo Clemente, Paul Hunter, Guillermo~A. P{\'{e}}rez,
  Mickael Randour, Jean{-}Fran{\c{c}}ois Raskin, Ocan Sankur, and Mathieu
  Sassolas.
\newblock Non-zero sum games for reactive synthesis.
\newblock In Adrian{-}Horia Dediu, Jan Janousek, Carlos Mart{\'{\i}}n{-}Vide,
  and Bianca Truthe, editors, {\em Language and Automata Theory and
  Applications - 10th International Conference, {LATA} 2016, Prague, Czech
  Republic, March 14-18, 2016, Proceedings}, volume 9618 of {\em Lecture Notes
  in Computer Science}, pages 3--23. Springer, 2016.
\newblock \href {http://dx.doi.org/10.1007/978-3-319-30000-9\_1}
  {\path{doi:10.1007/978-3-319-30000-9\_1}}.

\bibitem{DBLP:journals/corr/BruyereHR16}
V{\'{e}}ronique Bruy{\`{e}}re, Quentin Hautem, and Mickael Randour.
\newblock Window parity games: an alternative approach toward parity games with
  time bounds.
\newblock In Domenico Cantone and Giorgio Delzanno, editors, {\em Proceedings
  of the Seventh International Symposium on Games, Automata, Logics and Formal
  Verification, GandALF 2016, Catania, Italy, 14-16 September 2016.}, volume
  226 of {\em {EPTCS}}, pages 135--148, 2016.
\newblock \href {http://dx.doi.org/10.4204/EPTCS.226.10}
  {\path{doi:10.4204/EPTCS.226.10}}.

\bibitem{DBLP:conf/concur/BruyereHR16}
V{\'{e}}ronique Bruy{\`{e}}re, Quentin Hautem, and Jean{-}Fran{\c{c}}ois
  Raskin.
\newblock On the complexity of heterogeneous multidimensional games.
\newblock In Jos{\'{e}}e Desharnais and Radha Jagadeesan, editors, {\em 27th
  International Conference on Concurrency Theory, {CONCUR} 2016, August 23-26,
  2016, Qu{\'{e}}bec City, Canada}, volume~59 of {\em LIPIcs}, pages
  11:1--11:15. Schloss Dagstuhl - Leibniz-Zentrum fuer Informatik, 2016.
\newblock \href {http://dx.doi.org/10.4230/LIPIcs.CONCUR.2016.11}
  {\path{doi:10.4230/LIPIcs.CONCUR.2016.11}}.

\bibitem{DBLP:conf/emsoft/ChakrabartiAHS03}
Arindam Chakrabarti, Luca de~Alfaro, Thomas~A. Henzinger, and Mari{\"e}lle
  Stoelinga.
\newblock Resource interfaces.
\newblock In Rajeev Alur and Insup Lee, editors, {\em EMSOFT}, volume 2855 of
  {\em Lecture Notes in Computer Science}, pages 117--133. Springer, 2003.

\bibitem{DBLP:journals/tcs/ChatterjeeD12}
Krishnendu Chatterjee and Laurent Doyen.
\newblock Energy parity games.
\newblock {\em Theor. Comput. Sci.}, 458:49--60, 2012.
\newblock \href {http://dx.doi.org/10.1016/j.tcs.2012.07.038}
  {\path{doi:10.1016/j.tcs.2012.07.038}}.

\bibitem{DBLP:journals/iandc/Chatterjee0RR15}
Krishnendu Chatterjee, Laurent Doyen, Mickael Randour, and
  Jean{-}Fran{\c{c}}ois Raskin.
\newblock Looking at mean-payoff and total-payoff through windows.
\newblock {\em Inf. Comput.}, 242:25--52, 2015.
\newblock \href {http://dx.doi.org/10.1016/j.ic.2015.03.010}
  {\path{doi:10.1016/j.ic.2015.03.010}}.

\bibitem{DBLP:conf/lics/ChatterjeeHJ05}
Krishnendu Chatterjee, Thomas~A. Henzinger, and Marcin Jurdzinski.
\newblock Mean-payoff parity games.
\newblock In {\em 20th {IEEE} Symposium on Logic in Computer Science {(LICS}
  2005), 26-29 June 2005, Chicago, IL, USA, Proceedings}, pages 178--187.
  {IEEE} Computer Society, 2005.
\newblock \href {http://dx.doi.org/10.1109/LICS.2005.26}
  {\path{doi:10.1109/LICS.2005.26}}.

\bibitem{DBLP:conf/fossacs/ChatterjeeHP07}
Krishnendu Chatterjee, Thomas~A. Henzinger, and Nir Piterman.
\newblock Generalized parity games.
\newblock In Helmut Seidl, editor, {\em Foundations of Software Science and
  Computational Structures, 10th International Conference, {FOSSACS} 2007, Held
  as Part of the Joint European Conferences on Theory and Practice of Software,
  {ETAPS} 2007, Braga, Portugal, March 24-April 1, 2007, Proceedings}, volume
  4423 of {\em Lecture Notes in Computer Science}, pages 153--167. Springer,
  2007.
\newblock \href {http://dx.doi.org/10.1007/978-3-540-71389-0\_12}
  {\path{doi:10.1007/978-3-540-71389-0\_12}}.

\bibitem{DBLP:journals/acta/ChatterjeeRR14}
Krishnendu Chatterjee, Mickael Randour, and Jean{-}Fran{\c{c}}ois Raskin.
\newblock Strategy synthesis for multi-dimensional quantitative objectives.
\newblock {\em Acta Inf.}, 51(3-4):129--163, 2014.
\newblock \href {http://dx.doi.org/10.1007/s00236-013-0182-6}
  {\path{doi:10.1007/s00236-013-0182-6}}.

\bibitem{DBLP:conf/lics/DziembowskiJW97}
Stefan Dziembowski, Marcin Jurdzinski, and Igor Walukiewicz.
\newblock How much memory is needed to win infinite games?
\newblock In {\em Proceedings, 12th Annual {IEEE} Symposium on Logic in
  Computer Science, Warsaw, Poland, June 29 - July 2, 1997}, pages 99--110.
  {IEEE} Computer Society, 1997.
\newblock \href {http://dx.doi.org/10.1109/LICS.1997.614939}
  {\path{doi:10.1109/LICS.1997.614939}}.

\bibitem{EM79}
Andrzej Ehrenfeucht and Jan Mycielski.
\newblock Positional strategies for mean payoff games.
\newblock {\em Int. Journal of Game Theory}, 8(2):109--113, 1979.

\bibitem{DBLP:conf/mfcs/GimbertZ04}
Hugo Gimbert and Wieslaw Zielonka.
\newblock When can you play positionally?
\newblock In Jir{\'{\i}} Fiala, V{\'{a}}clav Koubek, and Jan Kratochv{\'{\i}}l,
  editors, {\em Mathematical Foundations of Computer Science 2004, 29th
  International Symposium, {MFCS} 2004, Prague, Czech Republic, August 22-27,
  2004, Proceedings}, volume 3153 of {\em Lecture Notes in Computer Science},
  pages 686--697. Springer, 2004.
\newblock \href {http://dx.doi.org/10.1007/978-3-540-28629-5\_53}
  {\path{doi:10.1007/978-3-540-28629-5\_53}}.

\bibitem{DBLP:conf/concur/GimbertZ05}
Hugo Gimbert and Wieslaw Zielonka.
\newblock Games where you can play optimally without any memory.
\newblock In Mart{\'{\i}}n Abadi and Luca de~Alfaro, editors, {\em {CONCUR}
  2005 - Concurrency Theory, 16th International Conference, {CONCUR} 2005, San
  Francisco, CA, USA, August 23-26, 2005, Proceedings}, volume 3653 of {\em
  Lecture Notes in Computer Science}, pages 428--442. Springer, 2005.
\newblock \href {http://dx.doi.org/10.1007/11539452\_33}
  {\path{doi:10.1007/11539452\_33}}.

\bibitem{DBLP:conf/dagstuhl/2001automata}
Erich Gr{\"{a}}del, Wolfgang Thomas, and Thomas Wilke, editors.
\newblock {\em Automata, Logics, and Infinite Games: {A} Guide to Current
  Research [outcome of a Dagstuhl seminar, February 2001]}, volume 2500 of {\em
  Lecture Notes in Computer Science}. Springer, 2002.

\bibitem{DBLP:conf/icalp/JurdzinskiLS15}
Marcin Jurdzinski, Ranko Lazic, and Sylvain Schmitz.
\newblock Fixed-dimensional energy games are in pseudo-polynomial time.
\newblock In Magn{\'{u}}s~M. Halld{\'{o}}rsson, Kazuo Iwama, Naoki Kobayashi,
  and Bettina Speckmann, editors, {\em Automata, Languages, and Programming -
  42nd International Colloquium, {ICALP} 2015, Kyoto, Japan, July 6-10, 2015,
  Proceedings, Part {II}}, volume 9135 of {\em Lecture Notes in Computer
  Science}, pages 260--272. Springer, 2015.
\newblock \href {http://dx.doi.org/10.1007/978-3-662-47666-6\_21}
  {\path{doi:10.1007/978-3-662-47666-6\_21}}.

\bibitem{DBLP:conf/icalp/Kopczynski06}
Eryk Kopczynski.
\newblock Half-positional determinacy of infinite games.
\newblock In Michele Bugliesi, Bart Preneel, Vladimiro Sassone, and Ingo
  Wegener, editors, {\em Automata, Languages and Programming, 33rd
  International Colloquium, {ICALP} 2006, Venice, Italy, July 10-14, 2006,
  Proceedings, Part {II}}, volume 4052 of {\em Lecture Notes in Computer
  Science}, pages 336--347. Springer, 2006.
\newblock \href {http://dx.doi.org/10.1007/11787006\_29}
  {\path{doi:10.1007/11787006\_29}}.

\bibitem{martin_AM75}
Donald~A. Martin.
\newblock Borel determinacy.
\newblock {\em Annals of Mathematics}, 102(2):363--371, 1975.

\bibitem{rECCS}
Mickael Randour.
\newblock Automated synthesis of reliable and efficient systems through game
  theory: A case study.
\newblock In {\em Proc. of ECCS 2012}, Springer Proceedings in Complexity XVII,
  pages 731--738. Springer, 2013.
\newblock \href {http://dx.doi.org/10.1007/978-3-319-00395-5\_90}
  {\path{doi:10.1007/978-3-319-00395-5\_90}}.

\bibitem{DBLP:journals/iandc/RouxP18}
St{\'{e}}phane~Le Roux and Arno Pauly.
\newblock Extending finite-memory determinacy to multi-player games.
\newblock {\em Inf. Comput.}, 261(Part):676--694, 2018.
\newblock \href {http://dx.doi.org/10.1016/j.ic.2018.02.024}
  {\path{doi:10.1016/j.ic.2018.02.024}}.

\bibitem{fsttcs}
Stéphane~Le Roux, Arno Pauly, and Mickael Randour.
\newblock Extending finite-memory determinacy by boolean combination of winning
  conditions.
\newblock In Sumit Ganguly and Paritosh Pandya, editors, {\em 38th IARCS Annual
  Conference on Foundations of Software Technology and Theoretical Computer
  Science, {FSTTCS} 2018, December 11--13, 2018, Ahmedabad, India}, volume 122
  of {\em LIPIcs}, pages 38:1--38:21. Schloss Dagstuhl - Leibniz-Zentrum fuer
  Informatik, 2018.

\bibitem{DBLP:conf/fossacs/Velner15}
Yaron Velner.
\newblock Robust multidimensional mean-payoff games are undecidable.
\newblock In Andrew~M. Pitts, editor, {\em Foundations of Software Science and
  Computation Structures - 18th International Conference, FoSSaCS 2015, Held as
  Part of the European Joint Conferences on Theory and Practice of Software,
  {ETAPS} 2015, London, UK, April 11-18, 2015. Proceedings}, volume 9034 of
  {\em Lecture Notes in Computer Science}, pages 312--327. Springer, 2015.
\newblock \href {http://dx.doi.org/10.1007/978-3-662-46678-0\_20}
  {\path{doi:10.1007/978-3-662-46678-0\_20}}.

\bibitem{DBLP:journals/iandc/VelnerC0HRR15}
Yaron Velner, Krishnendu Chatterjee, Laurent Doyen, Thomas~A. Henzinger,
  Alexander~Moshe Rabinovich, and Jean{-}Fran{\c{c}}ois Raskin.
\newblock The complexity of multi-mean-payoff and multi-energy games.
\newblock {\em Inf. Comput.}, 241:177--196, 2015.
\newblock \href {http://dx.doi.org/10.1016/j.ic.2015.03.001}
  {\path{doi:10.1016/j.ic.2015.03.001}}.

\bibitem{DBLP:journals/tcs/Zielonka98}
Wieslaw Zielonka.
\newblock Infinite games on finitely coloured graphs with applications to
  automata on infinite trees.
\newblock {\em Theor. Comput. Sci.}, 200(1-2):135--183, 1998.
\newblock \href {http://dx.doi.org/10.1016/S0304-3975(98)00009-7}
  {\path{doi:10.1016/S0304-3975(98)00009-7}}.

\end{thebibliography}

\end{document}